\documentclass[5p]{elsarticle}
\usepackage[utf8]{inputenc}
\usepackage{amsfonts}
\usepackage{fancyhdr}
\usepackage{comment}

\usepackage{times}
\usepackage{amsmath}
\usepackage{changepage}
\usepackage{amssymb}
\usepackage{subfig}
\usepackage{graphicx}
\usepackage{floatflt}
\usepackage{verbatim}
\usepackage{bbm}

\usepackage{tabularx}
\usepackage[dvipsnames]{xcolor}
\usepackage{colortbl}
\usepackage{bm}

\usepackage{xcolor}

\usepackage{url}
\usepackage{comment}
\usepackage{algorithm}
\usepackage{algpseudocode}

\usepackage[framemethod=tikz]{mdframed}
\usetikzlibrary{arrows,intersections,chains}
\usetikzlibrary{positioning}

\newtheorem{theorem}{Theorem}

\newtheorem{assumption}[theorem]{Assumption}

\newtheorem{definition}[theorem]{Definition}

\newtheorem{lemma}[theorem]{Lemma}

\newtheorem{proposition}[theorem]{Proposition}
\newtheorem{remark}[theorem]{Remark}

\newenvironment{proof}[1][Proof]{\textbf{#1.} }{\ \rule{0.5em}{0.5em}}

\DeclareMathOperator{\sgn}{sgn}

\newcommand{\tcb}{\textcolor{blue}}
\newcommand{\mc}{\mathcal}

\makeatletter
\def\ps@pprintTitle{%
   \let\@oddhead\@empty
   \let\@evenhead\@empty
   \let\@oddfoot\@empty
   \let\@evenfoot\@oddfoot
}
\makeatother

\colorlet{color1}{violet}
\colorlet{color2}{blue}
\colorlet{color_0}{color1!43.08249312026406!color2}
\colorlet{color_1}{color1!85.18349733898451!color2}
\colorlet{color_2}{color1!90.42204698692741!color2}
\colorlet{color_3}{color1!118.81321445593733!color2}
\colorlet{color_4}{color1!83.73329350241218!color2}
\colorlet{color_5}{color1!18.39962818297274!color2}
\colorlet{color_6}{color1!7.530786349233077!color2}
\colorlet{color_7}{color1!43.24842019105645!color2}
\colorlet{color_8}{color1!128.49214263497893!color2}
\colorlet{color_9}{color1!64.74311022346635!color2}
\colorlet{color_10}{color1!26.141419750395038!color2}
\colorlet{color_11}{color1!63.67653239513009!color2}
\colorlet{color_12}{color1!73.67580916883811!color2}

\tikzset{
    subnodesIEEE/.pic={
       
\draw [name path= line1, line width=1mm, color_0 ] (1.75,-8.5)-- (3.75,-8.5)  node[below ,midway] (anchor1) {} 
node[above ,left=0.8cm] (anchor1v2) {}
node[above ,left=0.4cm] (anchor1v5) {}
node[above ,left] (anchor1v7) {}
node[above ,right=0.2cm]  {\textcolor{black}{0}};

\draw [name path= line2, line width=1mm, color_1 ] (8.1,-15.2) node[below , right=0.1cm] (anchor2) {}-- (11.7,-15.2) 
    node[above ,right=0.2cm]  {\textcolor{black}{1}}
   node[below ,left=1cm] (anchor2v1) {} 
      node[below ,left=0.2cm] (anchor2v3) {} 
            node[below ,midway] (anchor2v6) {} 
            node[below ,left=0.4cm] (anchor2v4) {}  ;

\draw [name path= line3, line width=1mm, color_2 ] (13.9,-15.2) node[below] (anchor3) {} -- (16.75,-15.2) node[above ,right=0.2cm]  {\textcolor{black}{2}}
node[below ,left=1cm] (anchor3v2) {} 
node[below ,midway] (anchor3v4) {};

\draw [name path= line4, line width=1mm, color_3 ] (14.1,-11.5)-- (16.4,-11.5) 
node[above ,right=0.2cm]  {\textcolor{black}{3}}
node[above ,left=0.7cm] (anchor4v2) {} 
node[above ,left=0.5cm] (anchor4v3) {}
node[above ,left=0.9cm] (anchor4v5) {}
node[above ,left=0.25cm] (anchor4v7) {}
node[above, left=0.2cm] (anchor4) {};

\draw [name path= line5, line width=1mm, color_4 ] (8.8,-11.5)-- (10.8,-11.5)
node[above ,right=0.2cm]  {\textcolor{black}{4}} 
node[above ,left=0.7cm] (anchor5v1) {}
node[above ,left=0.4cm] (anchor5v2) {}
node[above ,left=0.2cm] (anchor5v4) {};

\draw [name path= line6, line width=1mm, color_5 ] (8.8,-6.4)-- (10.8,-6.4)
node[above ,right=0.2cm]  {\textcolor{black}{5}}
node [above, left=0.15cm] (anchor6) {}
node[above, left=0.9cm] (anchor6v1) {} 
node[above, left=0.4cm] (anchor6v2) {}
node[above, left=0.3cm] (anchor6v11) {}
node[above, left=0.6cm] (anchor6v13) {};

\draw [name path= line7, line width=1mm, black ] (14.8,-8.2)-- (16.2,-8.2) 
node[above ,right=0.2cm]  {\textcolor{black}{6}}
node[above, left=0.2cm] (anchor7v4) {} 
node[above, left=0.5cm] (anchor7v1) {}
node[above, left=0.1cm] (anchor7v8) {};

\draw [name path= line8, line width=1mm, color_7 ] (17.4,-7.8)-- (17.4,-6.3) 
node[above =0.2cm,]  {\textcolor{black}{7}}
node[above, midway] (anchor8v7) {}
node[above, left] (anchor8) {} ;

\draw [name path= line9, line width=1mm, color_8 ] (14.1,-6)-- (16.1,-6) 
node[above ,right=0.2cm]  {\textcolor{black}{8}}
node[above, left=0.6cm] (anchor9){}
node[above, left=0.8cm] (anchor9v4) {}
node[above, left=0.5cm] (anchor9v7) {}
node[above, left=0.2cm] (anchor9v14) {}  ;

\draw [name path= line10, line width=1mm, color_9 ] (13.3,-3.9)-- (14.7,-3.9)
node[above ,right=0.0cm]  {\textcolor{black}{9}}
node[above, left=0.6cm] (anchor10){}
node[above, left=0.2cm] (anchor10v9){}
node[above, midway] (anchor10v11) {};

\draw [name path= line11, line width=1mm, color_10 ] (9.9,-3.9)-- (11.2,-3.9) 
node[above ,right=0.2cm]  {\textcolor{black}{10}}
node[above, left= 0.2cm] (anchor11v10) {}

;
\draw [name path= line12, line width=1mm, color_11 ] (6.3,-3.9)-- (7.6,-3.9) 
node[above ,right=0.2cm]  {\textcolor{black}{11}}
node [above, midway] (anchor12) {} 
node [above, left=0.2cm] (anchor12v13) {}  ;

\draw [name path= line13, line width=1mm, color_12 ] (8.6,-2.5)-- (10.6,-2.5)
node[above ,right=0.2cm]  {\textcolor{black}{12}}
node[above, midway] (anchor13){}
node[above, left=0.4cm] (anchor13v6) {} 
node[above, left=0.2cm] (anchor13v14) {} 
node[above, left=0.8cm] (anchor13v12) {} ;

\draw [name path= line14, line width=1mm, color_6 ] (13.9,-2.5)-- (15.9,-2.5)
node[above ,right=0.2cm]  {\textcolor{black}{13}}
node[above, midway] (anchor14){}
node[above, left=0.2cm] (anchor14v9) {}
node[above, left=0.7cm] (anchor14v13) {} ;

\def\nd{0.3}
\def \idn{0.15} 
\node[draw,circle,fill=none,font=\small,line width=0.5pt, above= \nd cm of anchor1, node distance= 0.2cm] (G0) {G};

\draw (G0.south)-| (anchor1) ;

\node[draw,circle,fill=none,line width=0.5pt, below= \nd cm of anchor2] (c2) {C};

\draw (c2)-| (anchor2) ;

\node[draw,circle,fill=none,line width=0.5pt, below= \nd cm of anchor3] (c3) {C};
\draw (c3)-| (anchor3) ;
\node[draw,circle,fill=none,line width=0.5pt, right= \idn cm of c3] (g6) {G};
\draw (g6)--++(0,\nd cm) |- (anchor3v4) ;

\node[draw,circle,fill=none,line width=0.5pt, right= \idn cm of g6] (g7) {G};
\draw (g7)--++(0,\nd cm) |- (anchor3v4) ;

\node[draw,circle,fill=none,line width=0.5pt, below= \nd cm of anchor4] (c4) {C};
\draw (c4)--++(0,\nd cm) |- (anchor4) ;

\node[draw,circle,fill=none,line width=0.5pt, above= \nd cm of anchor5v1] (c5) {C};
\draw (c5.south)--++(0,\nd cm) |- (anchor5v1) ;

\node[draw,circle,fill=none,line width=0.5pt, right= \nd cm of anchor8] (g10) {G};
\draw (g10)--   (anchor8) ;

\node[draw,circle,fill=none,line width=0.5pt, below= \nd cm of g10] (g11) {C};

\draw (g11)-| (anchor8) ;

\node[draw,circle,fill=none,line width=0.5pt, above = \nd cm of anchor9] (c9) {C};
\draw (c9)--   (anchor9) ;
\node[draw,circle,fill=none,line width=0.5pt, above = \nd cm of anchor14] (c14) {C};
\draw (c14)--   (anchor14) ;
\node[draw,circle,fill=none,line width=0.5pt, below = \nd cm of anchor10] (c10) {C};
\draw (c10)--   (anchor10) ;

\node[draw,circle,fill=none,line width=0.5pt, below= \nd cm of anchor6] (g15) {G};

\draw (g15)-| (anchor6) ;
\node[draw,circle,fill=none,line width=0.5pt, left= \idn cm of g15] (c6) {C};
\draw (c6)-- ++ (0, 1) |- (anchor6) ;
\node[draw,circle,fill=none,line width=0.5pt, below= \nd cm of anchor11v10] (c11) {C};
\draw (c11)-| (anchor11v10) ;
\node[draw,circle,fill=none,line width=0.5pt, above= \nd cm of anchor12] (c12) {C};
\draw (c12.south)-| (anchor12) ;
\node[draw,circle,fill=none,line width=0.5pt, above= \nd cm of anchor13] (c13) {C};
\draw (c13.south)-| (anchor13) ;
    }}

\begin{document}

\begin{frontmatter}
    \title{Privacy Impact on Generalized Nash Equilibrium in Peer-to-Peer Electricity Market}

    \author[1,2]{Ilia Shilov \fnref{fn1}}
    \fntext[fn1]{Corresponding author}
    \ead{ilia.shilov@inria.fr}
    \author[2]{ H\'el\`ene Le Cadre}
    \author[1]{Ana Busic}
    
    \address[1]{Inria Paris, DI ENS, CNRS, PSL University, France}
    \address[2]{VITO/EnergyVille, Thorpark 8310, Genk, Belgium}
    \begin{abstract}
        We consider a peer-to-peer electricity market, where agents hold private information that they might not want to share. The problem is modeled as a noncooperative communication game, which takes the form of a Generalized Nash Equilibrium Problem, where the agents determine their randomized reports to share with the other market players, while anticipating the form of the peer-to-peer market equilibrium. In the noncooperative game, each agent decides on the deterministic and random parts of the report, such that (a) the distance between the deterministic part of the report and the truthful private information is bounded and (b) the expectation of the privacy loss random variable is bounded. This allows each agent to change her privacy level. We characterize the equilibrium of the game, prove the uniqueness of the Variational Equilibria and provide a closed form expression of the privacy price. In addition, we provide a closed form expression to measure the impact of the privacy preservation caused by inclusion of random noise and deterministic deviation from agents’ true values. Numerical illustrations are presented on the 14-bus IEEE network.
    \end{abstract}
    

\end{frontmatter}

\section{Introduction}

The large-scale integration of  Distributed  Energy  Resources (DERs), the increasing share of Renewable Energy Source (RES) - based generators in the energy mix and the more proactive role of prosumers, have led to the evolution of electricity markets from centralized pool-based organizations to decentralized peer-to-peer market designs \cite{tushar}. Within this peer-to-peer electricity market, agents negotiate their energy procurement seeking to minimize their costs with respect to both individual and coupling constraints, while preserving a certain level of privacy \cite{lecadre}. The problem is modeled as a generalized Nash equilibrium problem (GNEP), parametrized in the privacy level, chosen by the agents.

Information sharing in the peer-to-peer market can improve agents' performance, but also may violate their privacy, leading to the disclosure of agent's private information \cite{xie}. This calls for the design of new communication mechanisms that capture the agents' ability to define the information they want to share (their report) with the other market participants, while preserving their privacy \cite{fioretto}. In many applications, this problem is usually addressed by including noise to the reports that the agents subsequently use to compute the market equilibrium \cite{lecadre}. However, this approach does not include the ability of the agents to act strategically on the values of their report.
Moreover, the question of the optimal noise distribution is crucial in such a framework \cite{murguia}. 

To analyse the market in presence of shared coupling constraints, we employ Generalized Nash Equilibrium (GNE) as solution concept \cite{kulkarni}, and a refinement of it, called Variational Equilibria (VE), assuming the shadow variables associated with the shared coupling constraints are aligned among the agents. In our proposed framework, agents compute GNE with respect to the constraints that bound (a) the distance between the deterministic deviation from the true values of the private information and (b) the Kullback-Leibler divergence, that measures the effect of the additive random noise included in the reports.

Game theoretic approaches integrating the prosumers' strategic behaviors in the peer-to-peer trading are considered in \cite{lecadre}, \cite{belgioioso2}. The economic dispatch in energy communities under different structures of communications is analysed 
in \cite{moret_pinson}, \cite{moret}.  The impact of privacy on an energy community was analyzed in the literature, e.g. in \cite{lecadre}, where 
the sensitive information and the noise added to agents' reports were considered as exogenous parameters. 
Using a prediction model, Fioretto et al. provide a privacy-preserving mechanism, to protect the information exchanged between the different market operators while guaranteeing their coordination \cite{fioretto}. 
 
The anticipation of the actions of the agents in our model is represented by the common knowledge of the {\it form of the solution}. This anticipation will be used in strategic behavior framework to compute the prosumers' optimal deviation in their private information reports.

Various definitions of privacy have been introduced in the data science literature \cite{goncalves}. Several information metrics: e.g., mutual information, entropy, Kullback-Leibler divergence, and Fisher information are used to quantify information release \cite{farokhi}, \cite{murguia}. Differential privacy (DP) was recently successfully applied to multi-energy market operations \cite{fioretto} and dynamical systems \cite{murguia}. DP relies on adding noise to the reports from predetermined distributions. In our model we also use the additive noise, but we relax the assumptions of DP mechanism and focus on the prosumers' ability to determine their noise distribution. It is done by bounding the the expectation of the {\it privacy loss random variable} \cite{balle}, which constitutes exactly the Kullback-Leibler divergence for the introduced privacy-preserving randomized mechanism. 

To analyse the market in presence of shared coupling constraints, we employ Generalized Nash Equilibrium (GNE) as solution concept \cite{ kulkarni}, and a refinement of it, called Variational Equilibria (VE) \cite{kulkarni, rosen}. We focus on a certain properties of the game, such as {\it aggregative} and {\it potential} structure. Different algorithms using such properties as strong/strict-monotonicity of the game operator for computing VE using decentralized or semi-decentralized structure for multi-agent equilibrium problems in generalized aggregative games has recently gained high research interest \cite{belgioioso}, \cite{paccagnan}.

\subsection{Contributions}

We relate the notion of privacy preservation resulting from the non-disclosure of the nominal demands and RES-based generations of the prosumers in \cite{lecadre}, to the privacy mechanism with the additive Gaussian noise, that allows each agent to control her privacy level. It is done, firstly, by choosing the deterministic value to report to other agents; secondly, by using the random noise added to that value. We quantify the impact of privacy on the prosumers' costs and provide an analytical expression of the market equilibria. In addition, we allow each agent to change her level of privacy and show the existence of the incentives for the prosumers to deviate from their true sensistive parameter values. We rely on the notion of strong monotonicity to prove the existence and uniqueness of the solution to our problem. Using Kullback-Leibler divergence, we measure the cost of privacy, caused by inclusion of the random noise. All the theoretical results are illustrated on the 14-bus IEEE network.

The organization of the rest of this paper is as follows: in Section \ref{sec: statement} we first describe the peer-to-peer electricity trading problem in Subsections \ref{subsec: preliminaries} and \ref{subsec: electricity trading problem}, which constitutes a basis for our communication game, that will be defined in Subsection \ref{subsec: communication game}. In Section \ref{sec: equilibrium problems} we provide the analytical expression of the GNE, prove the uniqueness of the VE of the game and provide an expression for the utility gap, caused by the introduction of the privacy. Theoretical results are illustrated on the 14-bus IEEE network in Section \ref{sec: numerical}.

\section{Statement of the problem}\label{sec: statement}

\subsection{Preliminaries}\label{subsec: preliminaries}
In distributed control systems there is a usual trade-off between privacy and cost: to obtain a better solution, each agent relies on the information of the other agents in the system, which they might not have incentives to provide.

Consider a single-settlement market for peer-to-peer electricity trading made of a set $\mc{N}$ of $N$ agents, each one of them being located in a node of a communication network, that is modeled as a graph $\mc{G}:=(\mc{N},E)$ where $E\subseteq \mc{N}\times\mc{N}$ is the set of communication links between the players. Let $\Omega_n$ be the set of nodes, player $n$ wants to trade electricity with. Being the interface node between the local electricity market and at the distribution level and the transmission power network, node 0 can communicate with any other nodes in $\Omega_0 := \mathcal{N} \setminus 0$. The graph $\mc{G}$ does not necessarily reflect the distribution power network constraints.

In this paper we focus on the privacy issues that arise after solving the peer-to-peer electricity trading problem, considered in \cite{lecadre}. 

Each agent $n$ chooses independently her bilateral trades $\boldsymbol{q}_n$ with agents she wants to trade electricity with, self-generation $G_n$ and flexible demand $D_n$, in order to minimize her cost function $\Pi_n$:
\begin{equation}\label{eq: original cost function}
    \begin{aligned}
        &\Pi_n(D_n,G_n,\boldsymbol{q}_n) := \underbrace{1/2 \cdot a_n G_n^2 + b_n G_n +d_n}_{C_n(G_n)} + \\
        &+\underbrace{\tilde{a}_n(D_n - D^{\ast}_n )^2 - \tilde{b}_n}_{U_n(D_n)} + \underbrace{\sum_{m \in \Omega_n,m\neq n} c_{nm} q_{mn}}_{\tilde{C}_n(\bm{q})},
    \end{aligned}
\end{equation}
where $a_n, b_n, d_n, \tilde{a}_n, \tilde{b}_n > 0$ and $D^{\ast}_n$ denotes the \textit{nominal demand} of agent $n$ \cite{lecadre}. Thus, the vector of agent $n$'s decision variables is $(D_n, G_n, \bm{q}_n)$, where $\bm{q}_n := (q_{mn})_{m \in \Omega_n}$ is the vector of the quantities exchanged between $n$ and $m$ in the direction from $m$ to $n$, $q_{mn}$, for all $m \in \Omega_n \setminus \{n \}$. We use the following convention: if $q_{mn} \geq 0$, then $n$ buys $q_{mn}$ from $m$, otherwise ($q_{mn} < 0$) $n$ sells $-q_{mn}$ to $m$. We let $Q_n$ denote the {\it net import} of agent $n$: $Q_n := \sum_{m \in \Omega_n} q_{mn}$.

Each agent computes trading cost $\tilde{C}_n(\bm{q})$ using $c_{nm}$ which might represent preferences measured through product differentiation prices \cite{lecadre}, \cite{sorin} on the possible trades with the neighbors, or taxes. The following condition on agent's trades called {\it trading reciprocity constraint} couples the decisions of two neighboring agents, ensuring for every node $m \in \Omega_n$ that $q_{mn} + q_{nm} = 0$.
Note, that this formulation of the coupling constraints differs from the one presented in \cite{lecadre}, as we use equality constraint in our model, instead of the inequality. That means that the energy surplus is not allowed in the electricity trading model. Let $\kappa_{nm} \in [0, +\infty )$ be the equivalent trading capacity between node $n$ and node $m$, such that $\kappa_{nm} = \kappa_{mn}$ and $\forall m \in \Omega_n$. This equivalent trading capacity is used to bound the trading flows such that $q_{mn} \leq \kappa_{mn}$.

Local supply and demand should satisfy the following balance equality in each node $n$ in $\mathcal{N}$: $D_n = G_n + \Delta G_n + \sum_{m \in \Omega_n}  q_{mn}$, where $\Delta G_n$ is the \textit{renewable energy sources (RES)-based generation} at node $n$, assumed to be non-flexible.

\subsection{Electricity trading problem}\label{subsec: electricity trading problem}
As it was discussed in the introduction, each agent holds some private information that takes the form of nominal demand $D^{\ast}_n$ and RES-based generation $\Delta G_n$, which she does not desire to reveal to the other agents in the system. In the further analysis, we assume $y_n := D^{\ast}_n - \Delta G_n$ to be the private information of agent $n$. We assume that the agents desire to solve the electricity trading problem endowed with the set of coupled constraints while not allowing the other agents to infer their values of $y_n$. We denote $\bm{x}_n := (D_n, G_n, \bm{q}_n)$ to be the vector that contains agent $n$'s decision variables and $\bm{x_{-n}}$ is the vector of the other agents' actions. We recall the optimization problem formulated in \cite{lecadre} for the clearing of the peer-to-peer electricity market.

\subsubsection{Peer-to-peer market design}
In the peer-to-peer setting the problem of the electricity trading takes the form of generalized Nash equilibrium problem i.e., a game where the feasible sets of the players depend on the other players’ actions. With the notation introduced above, it means that each agent solves the following optimization problem:
\begin{subequations}\label{original problem p2p}
\begin{align}
    \min_{\bm{x}_n} \hspace{1cm} &   \Pi_n (\bm{x}_n),\\
    s.t. \hspace{1cm} & \underline{G}_n \leq G_n \leq \overline{G}_n  &(\tcb{\underline{\mu}_n, \overline{\mu}_n}) \label{eq: original G bounds}\\
    & \underline{D}_n \leq D_n \leq \overline{D}_n  &(\tcb{\underline{\nu}_n, \overline{\nu}_n}) \label{eq: original D bounds}\\
    &q_{mn} + q_{nm} = 0 &(\tcb{\zeta_{nm}}) \label{eq: reciprocity constraint}\\
    &q_{mn} \leq \kappa_{mn} &(\tcb{\xi_{nm}})\\
    &D_n = G_n + \Delta G_n + \sum_{m \in \Omega_n}  q_{mn}  &(\tcb{\lambda_n}) \label{eq: original balance},
\end{align}
\end{subequations}
where the corresponding dual variables are placed in blue at the right of each constraint. Note that in \eqref{original problem p2p} the feasible set of the agent $n$ can be rewritten in a more compact form $\mathcal{C}_n(\bm{x_{-n}}) = \{ \bm{x}_n |\eqref{eq: original G bounds} - \eqref{eq: original balance} \text{ hold} \}$. This notation will be used later in the paper.

We introduce the following assumption to guarantee that the interface trading capacities are big enough to supply trading needs of all the agents and that differentiation prices are symmetric for trading with the root node.
\begin{assumption}\label{assumption: root node}
    We assume that there are large trading capacities from and to node 0 -- that is $\xi_{0n} = \xi_{n0} = 0\,\, \forall n \in \mathcal{N}$ and $c_{n0}=c_{0n}$ for all $n \in \mathcal{N}$. 
\end{assumption}

The following subsection describes the computation of $\tilde{C}_n(\bm{q})$ in the different setting for the differentiation prices.

\subsection{Computation of the trading cost}

Under the conditions of Assumption \ref{assumption: root node}, Proposition 8 in \cite{lecadre} states that:
    \begin{proposition}\label{proposition: trade saturation lecadre}
    For any couple of nodes $n \in N, m \in \Omega_n, m \not= n$ with asymmetric preferences (such as $c_{mn} > c_{nm}$ or $c_{mn} < c_{nm}$) imply that the node with the smaller preference for the other saturates the line.
    \end{proposition}

We focus on two opposite instances:
\paragraph{1. All $c_{nm}$ are homogeneous}
That means that $c_{nm} = c$ for all $n,m \in \mathcal{N}$. This case reflects the interpretation of $c_{nm}$ as the taxes for energy trading, that should be naturally non-discriminating among agents. In this case bilateral trade cost is given by:
\begin{equation}\label{eq: bilateral cost trivial case}
    \tilde{C}_n(\bm{q}_n) = c \cdot Q_n
\end{equation}

\paragraph{2. All $c_{nm}$ for $m,n \not= 0$ are heterogeneous}
This framework represents the case, when all $c_{nm}, m,n \not= 0$ are drawn from some continuous distribution (e.g. uniform). Under the Assumption \ref{assumption: root node} we are able to obtain the expressions for $\tilde{C}_n(\bm{q}_n)$ in this framework for agent $n$. Again using Proposition \ref{proposition: trade saturation lecadre} we have that $q_{n0} = Q_n - \sum_{m \in \Omega_n, m \not= 0} \kappa_{nm} \sgn(c_{mn}-c_{nm})$,
where $Q_n$ is obtained by combining \eqref{eq: original balance} and expressions for $D_n, G_n$. Thus, we are able to obtain the cost expressions for each agents $n$ directly:
\begin{proposition}\label{proposition: bilateral cost p2p}
    Bilateral trade costs for any agent $n \in \mathcal{N}$ in the network except root node 0 are given by
    \begin{equation}
        \begin{aligned}
            \tilde{C}_n(\bm{q}_n) &= c_{0n} \big[  Q_n - \sum_{m \in \Omega_n, m \not= 0} \kappa_{nm} \sgn(c_{mn}-c_{nm})\big]\\
            &+\sum_{k \in \Omega_n, k \not=0} c_{nk} \kappa_{nk} \sgn(c_{kn}-c_{nk}).
        \end{aligned}
    \end{equation}
    Bilateral costs for node 0 are expressed as
    \begin{equation}
        \tilde{C}_0(\bm{q}_0) = \sum_{n \in \Omega_0} c_{0n} \Big[\sum_{m \in \Omega_n, m \not= 0} \kappa_{nm} \sgn(c_{mn}-c_{nm}) -  Q_n \Big]
    \end{equation}
\end{proposition}
\begin{remark}
   We do not impose any condition on the ratio between the values of the coefficients $c_{nm}$. Choosing $c_{0n} < c_{mn}, \forall m,n \in \mathcal{N}$, we can ensure the preference for the local trades.
\end{remark}

\paragraph{3. Intermediate case}
To demonstrate the difficulties arising in the general case for computing bilateral trades, we consider the intermediate case, in which there exists \textbf{one} additional symmetric relation $c_{n'm'} = c_{m'n'}$ for $m',n' \not=0$. Thus, for this pair of nodes we have that
\begin{equation*}
    \begin{aligned}
        Q_{n'} &= q_{0 n'} + q_{m'n'} + \sum_{k\not=m' \in \Omega_{n'}} \kappa_{n'k} \sgn(c_{kn'}-c_{n'k})\\
        Q_{m'} &= q_{0 m'} + q_{n'm'} + \sum_{k\not=n' \in \Omega_{m'}} \kappa_{m'k} \sgn(c_{km'}-c_{m'k}),
    \end{aligned}
\end{equation*}
where $q_{m'n'} = -q_{n'm'}$, which gives us a system of two equations with three unknown variables $q_{0 n'}, q_{0 m'}, q_{m'n'}$. Writing the similar equation for every node $k \not= m',n',0$, we get $N-3$ equations with $N-3$ unknowns and adding the expression for $Q_0$ we obtain linear system with $N$ independent equations and $N$ unknown variables. It follows that adding even one symmetric relation leads to the system of $N$ equations with $N+1$ unknowns.

\subsubsection{On the link between electricity trading and communication game}

It is shown in \cite{lecadre}, that at the VE, agent $n$'s decision variables $\bm{x}^{\ast}_n$ depend on the dual variable $\lambda_n$, which, under the Assumption \ref{assumption: root node} is aligned across agents: $\lambda_n = \lambda_0, \forall n \in \mathcal{N}$, where $\lambda_0$ is the  \textit{uniform market clearing price}. The equilibrium expressions, provided in \cite{lecadre}, also hold for our model with equality constraint \eqref{eq: reciprocity constraint}. $\lambda_0$ depends on the private information $y_n$ of the agents. Formally, $\lambda_0$ is given by:
\begin{equation}\label{eq: lambda}
    \lambda_0 = \frac{\sum_{n} y_n + \sum_{n} \frac{b_n}{a_n} }{ \sum_n \Big( \frac{1}{2\tilde{a}_n} + \frac{1}{a_n}  \Big)}
\end{equation}
and the decision variables $D_n$ and $G_n$ are given at the equilibrium by the following expressions:
$D_n(\bm{y}) = D^{\ast}_n - \frac{1}{2\tilde{a}_n} \lambda_0 $, $G_n(\bm{y}) = -\frac{b_n}{a_n} + \frac{1}{a_n}\lambda_0$. The expression for $Q_n$ is obtained from the supply demand equality condition \eqref{eq: original balance}: $Q_n(\bm{y}) = D^{\ast}_n + \frac{b_n}{a_n} - (\frac{1}{a_n}+\frac{1}{2\tilde{a}_n})\lambda_0 - \Delta G_n$.

Thus, to solve \eqref{original problem p2p} each agent needs to compute the uniform market clearing price $\lambda_0$, which requires a knowledge of all the $(y_n)_n$ in the system. 
It leads to a question for each agent $n$ of how to determine the report of her private information, so that it has the minimal impact on her cost, while guaranteeing that the certain level of privacy is met. That is, each agent $n$ anticipates the form of the solution of the electricity trading problem at the equilibrium and determines the report $\tilde{y}_n $ of her private information, that she submits to the other agents in the system. 

In order to do so, each agent $n$ minimizes the difference between the cost of the problem with the modified values and the optimal solution of the problem \eqref{original problem p2p} with the truthful reports $\Pi_n^{\ast}$:
\begin{equation}\label{original privacy problem}
    \begin{aligned}
        \min_{\tilde{y}_n} \hspace{1cm} &   \mathbb E \Big[ \Pi_n(\tilde{y}_n, \bm{\tilde{y}}_{-n}) - \Pi^{\ast}_n \Big],\\
        s.t. \hspace{1cm} & \bm{x}^{\ast}_n(\bm{\tilde{y}}) \in  \mathcal{C}_n(\bm{x}^{\ast}_{-n}(\tilde{y})),
    \end{aligned}
\end{equation}
where the expectation is taken in order to account for both randomized and deterministic cases. Note that $\bm{x}^{\ast}_n$ depends on $\bm{\tilde{y}}$ because in the expressions for the decision variables $D_n(\cdot), G_n(\cdot)$ and $\bm{q}_n(\cdot)$ we use the reports $\bm{\tilde{y}}$ instead of the true values $\bm{y}$ as the input. Also note that $\Pi^{\ast}_n$ is a constant as it is calculated using true values of $\bm{y}$, thus it can be omitted from the objective function.

\begin{remark}\label{remark: timing}
   In \eqref{original privacy problem} we assume that the form of the electricity trading problem is known by all the agents in the system. It enables each agent to anticipate the \textbf{form of the solution} $x^{\ast}_n(\cdot),$ for all $n \in \mathcal{N}$ and thus, based on this form to decide on the optimal information $\tilde{y}_n, \forall n \in \mathcal{N}$ to report to the other agents \textbf{before} they actually obtain the solution of the electricity trading problem. Note that it differs from  \cite{fioretto}, as we take the form of the solution $\bm{x}^{\ast}$ of the GNEP as given.
\end{remark}

\subsection{Communication game}\label{subsec: communication game}
The report of the agent $n$ takes the form $\tilde{y}_n = \hat{y}_n + \varepsilon_n$. The first part of the report captures the ability of agent $n$ to act strategically on her report by determining the deterministic part $\hat{y}_n$ that solves the cost minimization problem. 
In the second part of the report, each agent implements a randomized mechanism $M(\cdot)$ by choosing the noise $\varepsilon_n$ to add to $\hat{y}_n$ in order to preserve a certain level of privacy. 

\subsubsection{Privacy loss definition}

First, we define an upper bounded distance as a symmetric {\it adjacency relation} $y_n \simeq \hat{y}_n$ for agent $n$: $ y_n \simeq \hat{y}_n \Longleftrightarrow d(y_n, \hat{y}_n) \leq \alpha_n$,
where $\alpha_n$ is chosen beforehand and reflects the amount of information agent $n$ desires to preserve \cite{chatzikokolakis}. 

\begin{definition}[Privacy loss]
    Given a randomized mechanism $M$, let $p_{M(y_n)}(z)$ denote the density of the random variable $Z = M(y_n)$. The privacy loss function of $M(\cdot)$ on a pair of $y_n \simeq \hat{y}_n$ is defined as $l_{M,y_n,\hat{y}_n}(z) = \log \left( p_{M(y_n)}(z)/p_{M(\hat{y}_n)}(z) \right)$
    The privacy loss random variable $L_{M,y_n,\hat{y}_n} := l_{M,y_n,\hat{y}_n}(Z)$ is the transformation of the output random variable $Z = M(y_n)$ by the function $l_{M,y_n,\hat{y}_n}$.
\end{definition}
We assume that each agent samples a Gaussian noise $\varepsilon_n \sim \mathcal{N}(0,\sigma^2_n)$, thus obtaining the report $\tilde{y}_n \sim \mathcal{N}(\hat{y}_n,\sigma^2_n)$.
When the Gaussian isotropic random noise is added to the deterministic value of the input, it is well-known that the privacy loss random variable is also Gaussian:
\begin{lemma}[\cite{balle}]
    The privacy loss $L_{M,y_n,\hat{y}_n}$ of a Gaussian output perturbation mechanism follows a distribution $\mathcal{N}(\eta, 2\eta)$, with $\eta = D^2/2\sigma^2$, where $D = ||y_n - \hat{y}_n||$.
\end{lemma}

\subsubsection{
A randomized mechanism for information reporting}
We aim to allow agents to be able to decide on the optimal noise added to their private information, by choosing the optimal variance $V_n$. For simplicity of notations, we denote $V_n := \sigma^2_n$.

First, each agent chooses the neighboring input $\hat{y}_n \simeq y_n$, on which she later implements $M(\cdot)$. It is reflected in the constraint \eqref{privacy: signal bounds}.
In the constraint \eqref{privacy: expectation bound}, the expectation of the privacy loss random variable measures the expected privacy loss of the mechanism $M(y_n)$ on the fixed private information $y_n,\hat{y}_n$. In other words, it shows, how much information can be extracted from the report $\tilde{y}_n$.  Note, that it is exactly the Kullback-Leibler divergence (or the relative entropy) between $M$'s output distributions on $y_n$ and $\hat{y}_n$.

Thus, to decide on the optimal value of the report $\tilde{y}_n$, each agent needs to solve the following optimization problem:
\begin{subequations}\label{privacy problem}
    \begin{align} 
        \min_{\hat{y}_n, V_n} \quad & \mathbb{E}_{\varepsilon_n \sim \mathcal{N}(0,V_n)}  \Big[ \Pi_n(\bm{\hat{y}}, \bm{\varepsilon}) \Big] \label{privacy: objective}\\ 
        s.t. \quad & \underline{G}'_n \leq \mathbb{E}_{\varepsilon_n \sim \mathcal{N}(0,V_n)}\big[ G_n(\bm{\tilde{y}}) \big] \leq \overline{G}'_n 
        &(\tcb{\underline{\mu}_n, \overline{\mu}_n})\label{privacy: G bounds}\\
        & \underline{D}'_n \leq \mathbb{E}_{\varepsilon_n \sim \mathcal{N}(0,V_n)} \big[ D_n(\bm{\tilde{y}}) \big] \leq \overline{D}'_n  &(\tcb{\underline{\nu}_n, \overline{\nu}_n})\label{privacy: D bounds}\\
        & (\hat{y}_n - y_n )^2 \leq \alpha^2_n & (\tcb{\underline{\gamma}_n}, \tcb{\overline{\gamma}_n})\label{privacy: signal bounds}\\
        & \mathbb{E} \big[ L_{M,y_n,\hat{y}_n} \big] \leq A_n &(\tcb{\underline{\beta}_n}, \tcb{\overline{\beta}_n}) \label{privacy: expectation bound}
    \end{align}
\end{subequations}
where $\underline{G}'_n = \underline{G}_n + \omega_{G_n}$, $\overline{G}'_n = \overline{G}_n - \omega_{G_n}$ and $\underline{D}'_n = \underline{D}_n + \omega_{D_n}$, $\overline{D}'_n = \overline{D}_n - \omega_{D_n}$, in which $\omega_{D_n}, \omega_{G_n} >0$ are introduced in order to account for the strictly feasible solutions of problem \eqref{original problem p2p}. In the numerical experiments we set $\omega_{D_n}, \omega_{G_n}$ to be a small, e.g. $10^{-3}$.

As it is shown below, the only term depending on the variance in the utility function of the agent $n$ is $\frac{B_n}{B}\sum_m V_m$. In the special case $y_n = \hat{y}_n$ for some $n$, where the constraint \eqref{privacy: expectation bound} $\frac{(\hat{y}_n - y_n)^2}{2V_n}\leq A_n$ holds for any $0< V_n < \infty$. The possible convention could be to exclude this constraint from the consideration, when $y_n = \hat{y}_n$ and set $V_n = 0$. 

The condition for the uniform market clearing price $\lambda_0$ to have a form given in \eqref{eq: lambda} is to have  zero total net import, i.e. $\sum_n Q_n = 0$. In the case a {\it fully coordinated mechanism} is implemented, i.e the local MO has an access to all the constraints and parameters of the agents and solves the problem in a centralized way, it is possible to oblige agents to align their reports $\tilde{\bm{y}}$ such that $\mathbb{E}\Big[ \sum_n Q_n(\bm{\tilde{y}}) \Big] = \sum_n(y_n - \hat{y}_n) = 0$. It follows that  $\sum_n \hat{y}_n=  \sum_n y_n$. So, when we compute $\lambda_0$ using $\tilde{\bm{y}}$ instead of $\bm{y}$, we obtain $\mathbb{E}\big[ \lambda_0(\tilde{\bm{y}}) \big] = \frac{1}{B} (\sum_{n} \hat{y}_n + \sum_{n} \frac{b_n}{a_n}) = \frac{1}{B}(\sum_n y_n  + \sum_n \frac{b_n}{a_n})$
Thus, the final market clearing price does not depend on the reports of the agents, which is formalized in the following statement:
\begin{proposition}\label{proposition: zero total net import}
    When the prosumers align their reports $\tilde{\bm{y}}$ so that the condition $\mathbb{E}\big[ \sum_n Q_n(\tilde{\bm{y}}_n )\big] = 0$ is met, then the uniform market clearing price $\lambda_0$ depends only on the true values of their initial parameters $\bm{y}$.
\end{proposition}

In the case a peer-to-peer communication mechanism is implemented, the sum of the net imports at each node might not be equal to zero. Indeed, agents might have incentives to violate this condition in order to decrease their costs. Thus, the condition $\sum_n Q_n = 0$ might not hold.

On the market level it is necessary for the condition of zero total net import to hold such that supply and demand balance each other in problem \eqref{original problem p2p} \cite{moret}. Also, note that non zero total net import $\sum_n Q_n \not= 0$ implies that there exists at least one pair of agents $(n,m)$ with $q_{nm} + q_{mn} \not =0$. Besides, this might cause the violation of the capacity condition $q_{nm} \leq \kappa_{nm}$. It means that the local Market Operator (MO) has to compensate the difference $\mathbb{E}\big[ \sum_n Q_n(\tilde{\bm{y}}_n )\big]$ caused by the lack of coordination in the agents' reports. In the case $\mathbb{E}\big[ \sum_n Q_n(\tilde{\bm{y}}_n )\big] \leq 0$, there is an energy surplus in the system, which can be sold by the MO (by the intermediate of an aggregator) to the wholesale market at price $p^{0}$. If $\mathbb{E}\big[ \sum_n Q_n(\tilde{\bm{y}}_n )\big] \geq 0$, then the MO (by the intermediate of an aggregator) has to buy the energy on the wholesale market at price $p^{0}$, which depends on the wholesale market price, in order to supply the system demand. 

When the constraints and the private information of the agents are not shared, the MO only knows the aggregate deviation $\sum_n (y_n - \hat{y}_n)$ thus penalties imposed on the agents depend on it and not on the personal deviation $y_n-\tilde{y}_n$ of the agent $n$. 

\begin{remark}
   For prosumers, imports/exports of energy from/to the community manager are possible at prices $p^{-}$/$p^{+}$ respectively such that $p^{+} \leq p^{0} \leq p^{-}$. To avoid  non-differentiability in the utility function, we let $p^{+} = p^{0} = p^{-}$.
\end{remark}
To compensate for the cost of buying the lack of energy at the local market level from the wholesale market, the MO imposes penalty to each prosumer that takes the form $P(\bm{\tilde{y}}) = \frac{p^{0}}{N} \sum_n (y_n - \tilde{y}_n)$. Note that in case of the excess of the production on the local market level, the prosumers will be equally reimbursed based on the surplus produced.
The division by $N$ is introduced in order to equally split the burden of the non zero total net import and mitigate the possible volatility of the price $p^0$.

\begin{assumption}\label{assumption: MO}
    A local MO ensures the compensation for the nonzero total net import. This implies that the formula for $\lambda_0$ in \eqref{eq: lambda} is used by all the prosumers to compute their decision variables. 
\end{assumption}

\begin{proposition}\label{proposition: privacy price}
    Dual variables $\underline{\beta}_n, \overline{\beta}_n$ for the constraint \eqref{privacy: expectation bound} can be interpreted as the {\it privacy price} for agent $n$ and are computed by the formula
    \begin{equation*}
        (\overline{\beta}_n + \underline{\beta}_n)^2 = \frac{B^2_n (\hat{y}_n - y_n)^2}{4 B^4}
    \end{equation*}
\end{proposition}
\begin{proof}
    Constraint \eqref{privacy: expectation bound} can be rewritten as follows, when we consider $V_n \not=0$ for all $n \in \mathcal{N}$: $\mathbb{E} \big[ L_{M,y_n, \hat{y}_n}\big] \leq A_n \Longleftrightarrow   (\hat{y}_n - y_n)^2 \leq 2V_n A_n$.
    In the following analysis, we denote $B_n:= \frac{1}{a_n} + \frac{1}{2\tilde{a}_n}$ and $B := \sum_n B_n$. The objective function \eqref{privacy: objective} of the agent $n$ depends linearly on the $V_n$, thus attaining the minimum with respect to this decision variable on the lower boundary of the feasible region. The lower boundary is given by the constraint \eqref{privacy: expectation bound}, from which we can conclude that $V_n = \frac{(\hat{y}_n-y_n)^2}{2A_n}$ for any given value of the decision variable $\hat{y}_n$ for any agent $n$. From the KKT conditions we have that $V_n = \frac{2B^4}{A_n B^2_n}(\overline{\beta}_n + \underline{\beta}_n)^2$, from which we obtain the expression for $(\overline{\beta}_n + \underline{\beta}_n)^2$.
    
    First, from the complementarity conditions we know that either of $\underline{\beta}_n, \overline{\beta}_n$ equals 0. Clearly, it is non-negative term that appears in the utility of the agent $n$ at the equilibrium. Thus, we can view $\underline{\beta}_n, \overline{\beta}_n$ as a {\it privacy price}.
\end{proof}

\begin{remark}
   The privacy price increases with respect to the distance between the truthful ($y_n$) and biased ($\hat{y}_n$) values of agent $n$'s private information.
\end{remark}

\section{Equilibrium problems}\label{sec: equilibrium problems}

\subsection{Aggregate game formulation}
 
Note that as $\lambda_0$ depends on the sum of $\sum_n \big( D^{\ast}_n - \Delta G_n \big)$, the objective function in \eqref{privacy problem} has an {\it aggregative game} structure, i.e. it depends on player $n$'s decision $\hat{y}_n$ and on the aggregate of the other agents' decisions. 

Below we provide the computations of the objective function of agents $\Pi_n(\tilde{y}_n,\hat{y}_{-n})$ both in (i) the fully coordinated mechanism and (ii) the peer-to-peer coordination mechanism.

To arrive to this closed form expression, we observe that $\tilde{y}_n \sim \mathcal{N}(\hat{y}_n,V_n)$. The sum of normal variables is a normal variable itself: $\sum_n \tilde{y}_n \sim \mathcal{N}(\sum_n \hat{y}_n, \sum_n V_n)$,
from where it follows that $( \sum_n \tilde{y}_n + \sum_n \frac{b_n}{a_n}) \sim \mathcal{N}(\sum_n \hat{y}_n + \sum_n \frac{b_n}{a_n}, \sum_n V_n)$. Using a formula for the second moment of the normal distribution, expression \eqref{eq: bilateral cost trivial case} and Proposition \ref{proposition: bilateral cost p2p}, we obtain the expression for the utility of the agents in cases (i) and (ii).

In the homogeneous differentiation price case $c_{nm}=c$, the cost function of agent $n$ is given by
\begin{equation*}
    \begin{aligned}
        \mathbb{E} \big[ \Pi_n(\bm{\tilde{y}}) \big] &=  \frac{B_n}{2B^2} \Big[ \Big(\sum_m \hat{y}_m + \sum_m \frac{b_m}{a_m}\Big)^2 + \sum_m V_m \Big]\\
        &+ c \big[D^{\ast}_n - \Delta G_n +\frac{b_n}{a_n} - \frac{B_n}{B}\Big( \sum_m \hat{y}_m + \sum_m \frac{b_m}{a_m} \Big) \big]\\
        & + \frac{p^0}{N}\sum_m (y_m - \hat{y}_m) - \frac{b^2_n}{2a_n} + d_n  - \tilde{b}_n,
\end{aligned}
\end{equation*}
Expression for the utility in the case when $c_{nm}$ are heterogeneous is similar, except that for $\tilde{C}_n(\bm{q}_n)$, we use the expressions from Proposition \ref{proposition: bilateral cost p2p}.

\subsection{GNE computation}
\subsubsection{$c_{nm}$ are homogeneous}

From the computations of the KKT conditions, we obtain that $\frac{B_n}{B^2} \sum_{m \in \mathcal{N}} \hat{y}_m + M'_n = 0, \quad \forall n \in \mathcal{N}$,
where $M'_n := \frac{B_n}{B^2}\sum_m \frac{b_m}{a_m} -\frac{p^0}{N} - \frac{c B_n}{B}  + \frac{1}{a_n B}(\overline{\mu}_n - \underline{\mu}_n) + \frac{1}{2\tilde{a}_n B}(\underline{\nu}_n - \overline{\nu}_n) + \overline{\gamma}_n - \underline{\gamma}_n + \overline{\beta}_n - \underline{\beta}_n$. 

\subsubsection{$c_{nm}$ are heterogeneous for $m \not= 0$}
Analogously, first order stationarity conditions for agents $n \in \mathcal{N}$ are given by $\frac{B_n}{B^2} \sum_{m \in \mathcal{N}} \hat{y}_m + M''_n = 0$, with $M''_n:= \frac{B_n}{B^2}\sum_n \frac{b_n}{a_n}  - c_{0n}\frac{B_n}{B} -\frac{p^0}{N} + \frac{1}{a_n B}(\overline{\mu}_n - \underline{\mu}_n) + \frac{1}{2\tilde{a}_n B}(\underline{\nu}_n - \overline{\nu}_n) + \overline{\gamma}_n - \underline{\gamma}_n + \overline{\beta}_n - \underline{\beta}_n$ and  $M''_0 = \frac{B_0}{B^2} \sum_n \frac{b_n}{a_n} + \frac{1}{a_0 B}(\overline{\mu}_0 - \underline{\mu}_0) + \sum_{n \not=0 }c_{0n}\frac{B_n}{B} -\frac{p^0}{N} + \frac{1}{2\tilde{a}_0 B}(\underline{\nu}_0 - \overline{\nu}_0) + \overline{\gamma}_0 - \underline{\gamma}_0 + \overline{\beta}_0 - \underline{\beta}_0$.

\subsection{Uniqueness of the Variational Equilibrium}
\begin{definition}
    An operator $F: K \subseteq \mathbb{R}^n \rightarrow \mathbb{R}^n$ is strongly monotone on the set $\hat{K} \subseteq K$ with monotonicity constant $\alpha >0$ if $(F(x) - F(y))^{\top} (x - y) \geq \alpha ||x-y||^2, \quad \forall x,y \in \hat{K}.$
    The operator is monotone if $\alpha = 0$
\end{definition}
In order to show the uniqueness of the VE of the problem \eqref{privacy problem}, we check if the operator 
\begin{equation}\label{eq: pseudo gradient}
    F(\bm{\hat{y}}, \bm{V}) : = \big[\nabla_n \mathbb{E}\big( \Pi_n  (\bm{\hat{y}}, \bm{V}) \big)\big]^N_{n=1}
\end{equation}
is strongly monotone. To do so, we use the following lemma:
\begin{lemma}[\cite{paccagnan}]
    A continuously differentiable operator $F : K \subseteq \mathbb{R}^n \rightarrow \mathbb{R}^n$ is $\alpha$-strongly monotone with monotonicity constant $\alpha$ (resp. monotone) if and only if $\nabla_x F(x) \succeq \alpha I$ (resp. $\nabla_x F(x) \succeq 0$) for all $x \in K$. Moreover, if $K$ is compact, then there exists $\alpha > 0$ such that $\nabla_x F(x) \succeq \alpha I$ for all $x \in K$ if an only if $\nabla_x F(x) \succ 0$ for all $x \in K$.
\end{lemma}
For homogeneous differentiation price $c_{nm} = c$, $\bm{F}(\bm{\hat{y}}, \bm{V})$ writes as follows:
\begin{multline}\label{eq: pseudo-differential}
    \bm{F}(\bm{\hat{y}}, \bm{V}) = col \Bigg\{ \Bigg( \frac{B_i}{B^2} \left(\sum_{m} \hat{y}_m + \sum_{m} \frac{b_m}{a_m}\right)-\\ -c\frac{B_i}{B} - \frac{p^0}{N},  \frac{V_i B_i}{B^2} \Bigg)^{N-1}_{i=0} \Bigg\}
\end{multline}
When differentiation prices $c_{nm}$ are heterogeneous for $m \not= 0$, operator $\bm{F}(\bm{\hat{y}}, \bm{V})$ is obtained similarly, but for the expressions of $\tilde{C}_n(\bm{q}_n)$, we take expressions  from Proposition \ref{proposition: bilateral cost p2p}.

\begin{lemma}\label{lemma: strong monotonicity}
    Operator $\bm{F}(\bm{\hat{y}}, \bm{V})$ defined in \eqref{eq: pseudo-differential} 
    is strongly monotone.
\end{lemma}

\begin{proof}
    The proof can be found in the Appendix. 
\end{proof}

\begin{proposition}
    By the strong monotonicity of $\bm{F}(\bm{\hat{y}}, \bm{V})$, VE of the game \eqref{privacy problem} is unique \cite{kulkarni}.
\end{proposition}

\subsection{Generalized Potential Game extension}

\begin{assumption}\label{assumption: equivalence of the B_n}
    Assume that $\forall i, j \in \mathcal{N}: \frac{B_i}{B} \simeq \frac{B_j}{B}$, i.e. each agent $n$'s contribution $B_n$ to the sum $B$ is relatively small. Denote $H:=\frac{B_n}{B} \forall n \in \mathcal{N}$.
\end{assumption}
\begin{proposition}\label{proposition: GPG}
    Under Assumption \ref{assumption: equivalence of the B_n}, the game \eqref{privacy problem} is a Generalized Potential Game.
\end{proposition}

\begin{proof}
    The proof can be found in the Appendix. 
\end{proof}

Generalized Potential Games constitute a subclass of games for which the convergence of the BR algorithms is established \cite{facchinei} in the deterministic case. Taking into account that the BR scheme is suited for our private framework, an interesting direction of the research would be to establish the convergence of the BR algorithm for the stochastic NE of the GPG. 

\section{Numerical Results}\label{sec: numerical}
\subsection{Algorithm description}
In the paper \cite{Yu}, authors employ the penalized individual cost functions to deal with coupled constraints and provide a three stochastic gradient strategies with constant step-sizes in order to approach the Nash Equilibrium. In order to establish their results, authors consider the model with the operator $F(\bm{\hat{y},\bm{V}})$ to be {\it strongly-monotone} and {\it Lipschitz continuous}, which holds for our case. We consider the scheme, called by the authors as {\it Diffusion Adapt-then-Penalize}:
\begin{equation*}
        \left\{
        \begin{aligned}
            &\psi^k_{\nu} = \hat{y}^{k-1}_{\nu} - \mu \nabla_{\hat{y}_{\nu}}\Pi_{\nu}(\hat{y}_{\nu}, \bm{\tilde{y}}_{-\nu})   \\
            &\hat{y}^k_{\nu} = \psi^{k}_{\nu} - \mu R \nabla_{\hat{y}_{\nu}}(\theta_{\nu}(\psi^k_{\nu})),
        \end{aligned}
        \right.
\end{equation*} 
where $\mu$ denotes the step-size, $R$ - penalty parameter and $\theta_{\nu}(\cdot)$ - penalty function for the coupling constraints, which we choose to be the sum over all the constraints of the form $x \leq 0$, of all functions $p_{\nu}(x)$ such that $p_{\nu} (x)= \sum_{} \mathbb{I}_{x \geq 0} \cdot x^2/2$.
\subsection{Numerics}

We consider the IEEE 14-bus network system, which is depicted in Figure \ref{fig:IEEE_network}, where each bus (node) of the network corresponds to a prosumer in our model. We consider the system, which consists of the agents with the non-zero self-generation and demand parameters, thus we exclude one interim node 6, which sole purpose in the initial 14-bus system is to connect the flows. Thus, in our model there are 13 buses (nodes). As parameters of the algorithm, we set $\mu = 0.003$ and $R = 700$. We first focus on the homogeneous differentiation price case $c_{nm} = 1.0\,$[\$/MWh], for all $n,m \in \mathcal{N}$. The cost $p^0$, used by local MO to trade with the wholesale market, is set to be higher than $c$ and equals $5.0\,$[\$/MWh]. The natural assumption is the homogeneity of the self-generation parameters of the prosumers, which we set to be $a_n = 0.5 , b_n = 6.0$ for all $n \in \mathcal{N}$. Also, there are three nodes (2, 5, 7) that are additionally equipped with a RES-based generation. Values on the links between the nodes on Figure \ref{fig:IEEE_network} specify the trading capacity parameters $\kappa_{nm}$. Recall from the Assumption \ref{assumption: root node} that there are large trading capacities to and from node 0, thus we do not specify them on the scheme. The nominal demands and RES-based generations in Figure \ref{fig:IEEE_network} are given in [GWh]. All the parameters that are used to calibrate the agents' utility functions are specified in Table \ref{table 1}.

\begin{table}[ht]
\caption{Agents' utility function parameters.}
\label{table 1}
\begin{center}
\scalebox{0.7}{\begin{tabular}{|c|c|c|c|c|c|c|c|}
\hline
        \textbf{Node}     & $\bm{\tilde{a}}$ & $\bm{\tilde{b}}$ & \textbf{d}  & $\overline{D}$ & $\overline{G}$  \\ \hline \hline
        \textbf{0}         & 1.5 & 0.0 & 9  & 25.0 & 100.0  \\ 
        \textbf{1}         & 1.18 & 5.09 & 15  & 26.7 &  100.0   \\ 
        \textbf{2}          & 1.0 & 3.78 & 14  & 99.2 &  80.0 \\ 
        \textbf{3}        & 0.57 & 4.36 & 0.0 & 52.8 & 20.0 \\ 
        \textbf{4}        & 1.24 & 5.03 & 0.0 & 12.6 & 20.0 \\ 
        \textbf{5}         & 1.62 & 3.04 & 2.0 & 16.2 & 20.0\\ 
        \textbf{6}         & 1.54 & 4.29 & 0.0 & 19.9 & 20.0\\ 
        \textbf{7}        & 1.5 & 0.0 & 11.0 & 25.0 & 50.0\\ 
        \textbf{8}         & 0.31 & 2.75 & 0.0 & 34.5 &20.0\\ 
        \textbf{9}         & 4.36 & 4.67 & 0.0 & 14.0 &20.0\\ 
        \textbf{10}        & 1.63 & 3.32 & 0.0 & 8.5 &20.0\\ 
        \textbf{11}         & 5.16 & 5.5 & 0.0 & 11.1 &20.0\\ 
        \textbf{12}          & 1.96 & 6.21 & 0.0 & 18.5 &20.0\\ 
        \hline
\end{tabular}}
\end{center}
\end{table}

\begin{figure}
\begin{center}
\begin{tikzpicture}[scale=0.5, >=latex,every node/.style={inner sep=0pt,outer sep=0,font=\small}]
\pic[scale=0.5, >=latex,every node/.style={inner sep=0pt,outer sep=0,font=\small}]{subnodesIEEE};

\tikzstyle{line1v2} = [draw,line width=1pt, color=black,text=black,sloped ]
\tikzstyle{line1v5} = [draw,line width=1pt, color=black,text=black,sloped ]
\tikzstyle{line1v6} = [draw,line width=1pt, color=black,text=black,sloped ]
\tikzstyle{line1v7} = [draw,line width=1pt, color=black,text=black,sloped ]
\tikzstyle{line1v12} = [draw,line width=1pt, color=black,text=black,sloped ]
\tikzstyle{line2v3} = [draw,line width=1pt, color=black,text=black,sloped ]
\tikzstyle{line2v4} = [draw,line width=1pt, color=black,text=black,sloped ]
\tikzstyle{line2v5} = [draw,line width=1pt, color=black,text=black,sloped ]
\tikzstyle{line3v4} = [draw,line width=1pt, color=black,text=black,sloped ]
\tikzstyle{line4v5} = [draw,line width=1pt, color=black,text=black,sloped ]
\tikzstyle{line4v7} = [draw,line width=1pt, color=black,text=black,sloped ]
\tikzstyle{line4v9} = [draw,line width=1pt, color=black,text=black,sloped ]
\tikzstyle{line5v6} = [draw,line width=1pt, color=black,text=black,sloped ]
\tikzstyle{line6v11} = [draw,line width=1pt, color=black,text=black,sloped ]
\tikzstyle{line6v12} = [draw,line width=1pt, color=black,text=black,sloped ]
\tikzstyle{line6v13} = [draw,line width=1pt, color=black,text=black,sloped ]
\tikzstyle{line7v8} = [draw,line width=1pt, color=black,text=black,sloped ]
\tikzstyle{line7v9} = [draw,line width=1pt, color=black,text=black,sloped ]
\tikzstyle{line9v10} = [draw,line width=1pt, color=black,text=black,sloped ]
\tikzstyle{line9v14} = [draw,line width=1pt, color=black,text=black,sloped ]
\tikzstyle{line10v11} = [draw,line width=1pt, color=black,text=black,sloped ]
\tikzstyle{line12v13} = [draw,line width=1pt, color=black,text=black,sloped ]
\tikzstyle{line13v14} = [draw,line width=1pt, color=black,text=black,sloped ]

\node [below=10pt of anchor4v2] (anchor4v2part) {};
\renewcommand{\=}{\hspace{-2pt} = \hspace{-2pt}}

\draw [line1v2] (anchor1v2) -- node [below, yshift=-.1cm] {} (anchor2v1);

\draw [line1v5] (anchor5v1) --++ (0,-30pt) --++ (-30pt,0) -- node [below,yshift=-.1cm] {} (anchor1v5) ;
\draw [line1v6] (anchor1v5) -- node [below,yshift=-.1cm] {} ++ (6,.8) --++ (0,1.3) ;
\draw [line1v7] (anchor1v7) --  ++ (50pt,-38pt) node [below,yshift=-.1cm,xshift= .5cm] {} -| (anchor7v1) ;
\draw [line1v12] (anchor1) -- node [below,yshift=-.1cm] {} (anchor12);
\draw [line2v3] (anchor2v3) --++ (0,20pt)-- node [below,yshift=-.1cm] {\tiny $36.0$}   ++(3.2,0) --++ (0,-20pt) ;
\draw [line2v4] (anchor2v4) --++ (0,20pt) -- node [below,yshift=-.1cm] {\tiny $65.0$}  (anchor4v2part) -- (anchor4v2) ;
\draw [line2v5] (anchor2v6) -- node [below,yshift=-.1cm] {\tiny $50.0$}   ++ (0,3.7);
\draw [line3v4] (anchor3v4) -- node [below,xshift=-0.2cm,yshift=-.1cm] {\tiny $65.0$}   ++(0,3.7);
\draw [line4v5] (anchor4v5) --++ (0,-20pt) -- node [below,yshift=-.1cm] {\tiny $45.0$}  ++(-4,0) --++ (0,20pt);
\draw [line4v7] (anchor4v7) --  node [below,yshift=-.1cm] {\tiny $32.0$} ++ (0,3.3) ;
\draw [line4v9] (anchor4v5) -- node [above,yshift=0.1 cm,xshift=0.2cm,] {\tiny $32.0$} ++ (0,5.5) ;
\draw [line5v6] (anchor5v2) -- node [below,yshift=-.1cm,xshift=-0.3cm,] {\tiny $45.0$}  ++(0,5);
\draw [line6v11] (anchor6v11) -- node [above,yshift=.1cm] {\tiny $18.0$} ++ (0,2.5 ) ;
\draw [line6v12] (anchor6v1) -- node [below,yshift=-.1cm] {\tiny $32.0$}  (anchor12);
\draw [line6v13] (anchor6v13) -- node [above,yshift=.1cm] {\tiny $32.0$} ++ (0,3.9) ;
\draw [line7v8] (anchor7v8) --++(0,35pt) -- node [below,yshift=-.1cm] {\tiny $32.0$}   ++(1.4,0);
\draw [line7v9] (anchor7v1) -- node [below,yshift=-.1cm] {\tiny $32.0$} ++ (0,2.2) ;
\draw [line9v10] (anchor9v4) -- node [above,yshift=.1cm] {\tiny $32.0$} ++ (0,2.1) ;
\draw [line9v14] (anchor9v14) -- node [below,xshift=0.2cm,yshift=-.1cm] {\tiny $12.0$} ++ (0,3.5) ;
\draw [line10v11] (anchor10v11) --++(0,20pt) -- node [below,yshift=-.1cm] {\tiny $12.0$}  ++(-3.5,0)  --++(0,-20pt) ;
\draw [line12v13] (anchor12v13) -- node [above,yshift=.1cm] {\tiny $12.0$}  (anchor13v12) ;
\draw [line13v14] (anchor13v14) --++(0,20pt) -- node [above,yshift=.1cm] {\tiny $12.0$}   ++ (3.8,0) --++(0,-20pt);

\def\nd{0.3}
\def \idn{0.15} 
\node [above=0.5cm of G0] (tg1) {\tiny Grid connection}; 
\node [ below=0.8cm of g6,rotate=90] (tg4) {\tiny \begin{tabular}{c}
     wind \\
     $\Delta G\=4.99$
\end{tabular}  }; 
\node [ below=0.5cm of g7,rotate=90] (tg7) {\tiny coal}; 
\node [ right=0.2cm of g10] (tg10) {\tiny \begin{tabular}{c}
     \kern-2em wind \\
     $\mkern-28mu \Delta G\=7.5$
\end{tabular}  }; 
\node [ right=0.5cm of g11] (tg11) {\tiny $\mkern-36mu D^{\ast}_7 = 4.0$}; 
\node [ below=0.5cm of g15,rotate=90] (tg15) {\tiny \begin{tabular}{c}
      solar \\
     $\Delta G\=15.51$
\end{tabular}  }; 

\node [ below=0.8cm of c2,rotate=90] (tc2) {\tiny $D^*_{2} \=6.57$} ;
\node [ below=0.8cm of c3,rotate=90] (tc3) {\tiny $D^*_{3} \=12.55$} ;
\node [right= 0.2cm of c4] (tc4) {\tiny $D^*_{4} \=8.75$} ;
\node [left= 0.2cm of c5] (tc5) {\tiny $D^*_{5} \=6.37$} ;
\node [ below=0.55cm of c6,xshift=-0.1cm,rotate=90] (tc6) {\tiny $D^*_{6} \=4.33$} ;
\node [right= 0.4cm of c9] (tc9) {\tiny $D^*_{9} \=9.42$} ;
\node [ below=0.8cm of c10,rotate=90] (tc10) {\tiny $D^*_{10} \=3.27$} ;
\node [below right= 0.1cm of c11] (tc11) {\tiny $D^*_{11} \=4.51$} ;
\node [above= 0.8cm of c12,rotate=90] (tc12) {\tiny $D^*_{12} \=3.26$} ;
\node [above= 0.8cm of c13,rotate=90] (tc13) {\tiny $D^*_{13} \=5.63$} ;
\node [above= 0.8cm of c14,rotate=90] (tc14) {\tiny $D^*_{14} \=5.28$} ;


\end{tikzpicture}

\caption{\small{IEEE 14-bus network system }}
\label{fig:IEEE_network}
\end{center}
\end{figure}
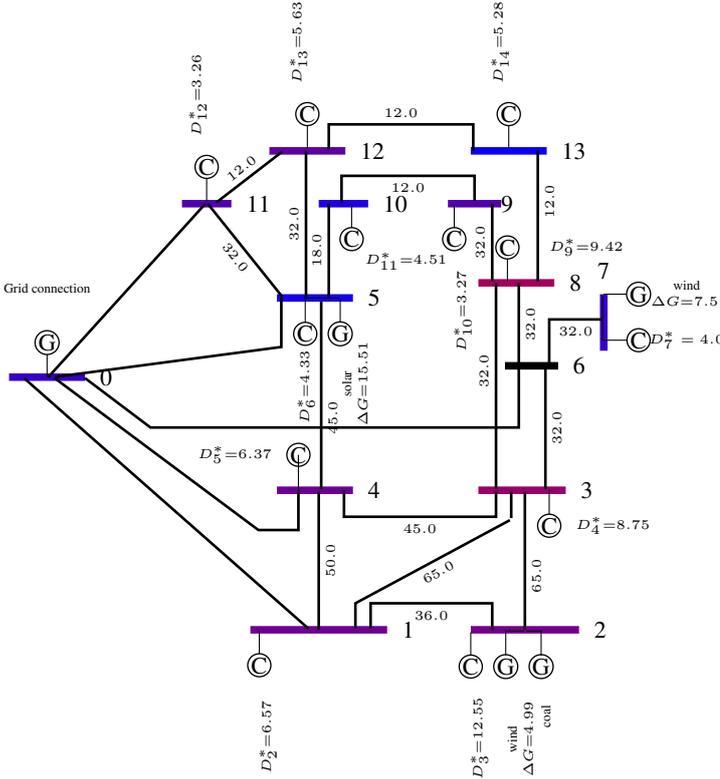

We measure the impact of our mechanism on the cost of the agents, measured through the agent's utility gap $\mathbb{E} \Big[\Pi^{\ast}_n - \Pi_n(\tilde{y}_n, \bm{\tilde{y}}_{-n}) \Big]$ for each agent $n$.  In Figure \ref{fig:A_n}, we plot the utility gap as a function of $A_n$ for all the agents. We observe that the nodes 3 and 8 decrease their costs the most among all the prosumers and nodes 9 and 11 have, on the contrary, increasing costs. From Table \ref{table 1} it can be seen that node 3 and 8 have the minimal flexible demand coefficients: $\tilde{a}_3 = 0.57$ and $\tilde{a}_8 = 0.31$. Similarly, nodes 9 and 11 have the biggest flexible demand coefficients: $\tilde{a}_9 = 4.36$ and $\tilde{a}_{11} = 5.16$. The cost of the demand flexibility affects the utility of the agents, i.e. small cost allows them to adjust their demand such that they can decrease their costs, while deviating from their true values.

For the graphs shown below, we set $\alpha_n = 3.0$ when we plot the dependance w.r.t. $A_n$, and $A_n = 10.0$ when we plot the dependance w.r.t. $\alpha_n$. For this choice of parameters, the color of the nodes in Figure \ref{fig:IEEE_network} shows the {\it privacy price} $\underline{\beta}_n, \overline{\beta}_n$ [\$/MWh] from Proposition \ref{proposition: privacy price} in each $n \in \mathcal{N}$. Light \textcolor{blue}{blue} denotes the lowest privacy price ($1.129.10^{-3}$ [\$/MWh]) and dark \textcolor{violet}{violet} denotes the highest ($2.827.10^{-2}$[\$/MWh]).

Figure \ref{fig:alpha} represents the dependence of the plot of the utility gap on the parameter $\alpha_n$ of the agents. It is shown, that when the maximal bound on the distance is low, the agents expectedly deviate from their costs $\Pi^{\ast}_n$. As soon as $\alpha_n$ increases, thus providing more possibility to deviate, agents tend to show the similar behavior as on the plot with respect to $A_n$: nodes 3 and 8 gain the most and nodes 9 and 11 have the increasing costs.
\begin{figure}[ht]
\centering
\begin{minipage}{.45\linewidth}
  \includegraphics[width=\linewidth]{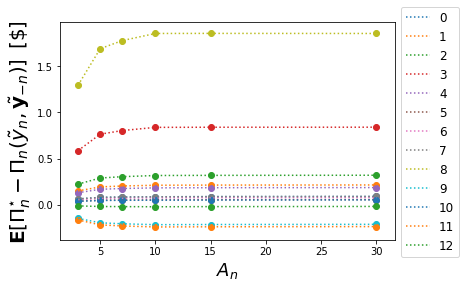}
  \caption{Utility gap wrt. $A_n$}
  \label{fig:A_n}
\end{minipage}
\hspace{.03\linewidth}
\begin{minipage}{.45\linewidth}
  \includegraphics[width=\linewidth]{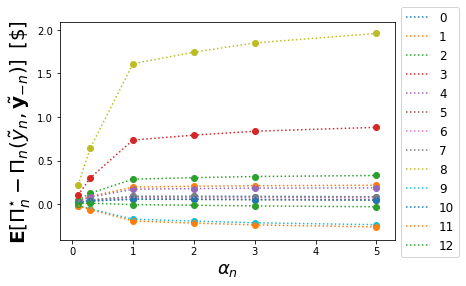}
  \caption{Utility gap wrt. $\alpha_n$}
  \label{fig:alpha}
\end{minipage}
\end{figure}

Figures \ref{fig:A_n social cost} and \ref{fig:alpha social cost} depict the dependance of the social cost of the system w.r.t. $A_n$ and $\alpha_n$ respectively. We compare three instances: peer-to-peer communication mechanism, fully coordinated communication mechanism and the social cost evaluated in the truthful reports. Note, that the latter one provides the same cost when Proposition \ref{proposition: zero total net import} holds. 

It can be seen that increase in $A_n$ affects the peer-to-peer communication the most, which is caused by the decrease of the privacy price induced by the noise in the agents' reports (recall, that the variance is given by $V_n = \frac{2B^4}{A_n B^2_n}(\overline{\beta}_n + \underline{\beta}_n)^2$), thus allowing them to compute their decision variables more precisely. Clearly, it affects the centralized communication mechanism less. On the other hand, increase of the $\alpha_n$ affects the centralized communication mechanism the most, as it allows the local MO to find an optimal solution for each agent in the system, thus leading to the biggest decrease in the costs. 
\begin{figure}
\centering
\begin{minipage}{.45\linewidth}
  \includegraphics[width=\linewidth]{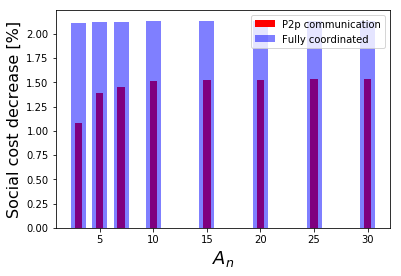}
  \caption{Social cost decrease wrt. $A_n$}
  \label{fig:A_n social cost}
\end{minipage}
\hspace{.05\linewidth}
\begin{minipage}{.45\linewidth}
  \includegraphics[width=\linewidth]{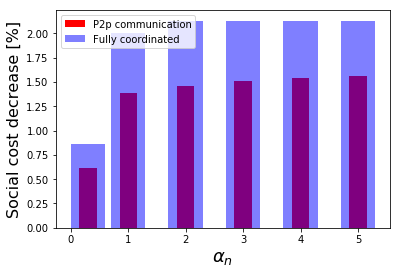}
  \caption{Social cost decrease wrt. $\alpha_n$}
  \label{fig:alpha social cost}
\end{minipage}
\end{figure}

The results above are shown in the homogeneous $A_n$, $\alpha_n$ and $c_{nm}$ case. Heterogeneity in the distance parameters $\alpha_n$ and $A_n$ does not affect the behavior of the agents in the system described above. Nevertheless, setting parameters $\alpha_n$ to be small for those who deviate the most (e.g. nodes 3 and 8) can bound their influence on the sum $\sum_n \tilde{y}_n$, thus, bounding the deviation from the $\sum_n y_n$. 

In the case of heterogeneous differentiation prices $c_{nm}$ for $n,m \not=0$, we compute the trading costs of the agents, using the expressions given in Proposition \ref{proposition: bilateral cost p2p}. Numerical experiments show the same behavior for all the agents in the system, while distinguishing the node 0: in this setting it decreases its cost the most. An interesting research direction would be to adapt the penalty for the prosumers in order to account for this behavior.

\section{Conclusion}
In our work we considered a peer-to-peer electricity market, in which agents have private information. The problem is modeled as a noncooperative communication game, which takes the form of a GNEP, where the agents determine their randomized reports to share with the other market players, while anticipating the form of the peer-to-peer market equilibrium. Agents decide on the deterministic and random parts of the report, such that the (a) the distance between the deterministic part of the report and the truthful private information is bounded (b) expectation of the privacy loss random variable is bounded. This allows them to act strategically on the values of the deterministic part and to choose the random noise included in their reports. We characterized the equilibrium of the problem and proved the uniqueness of the Variational Equilibria. We provided a closed form expression for the privacy price. The theoretical results are illustrated on the 14-bus IEEE network, using the stochastic gradient descent algorithm. We show the impact of the privacy preservation caused by inclusion of random noise and deterministic deviation from agents’ true values.

Since our problem has a potential form under mild assumptions, as the next step, we will focus on the development of the distributed learning algorithm for the stochastic NE of the Generalized Potential Game. Another interesting research direction would be to consider the decentralized communication mechanism, where agents do not have the ability to anticipate the form of the uniform market price.


\newpage
\section{Appendix}
\subsection{Proof of Lemma \ref{lemma: strong monotonicity}}
\begin{proof}
    First, note that for heterogeneous $c_{nm}$, operator $\bm{F}(\bm{\hat{y}}, \bm{V})$ writes as follows for nodes $i \not= 0$:
    \begin{multline}\label{eq: pseudo-differential heterogeneous}
        \bm{F}(\bm{\hat{y}}, \bm{V}) = col \Bigg\{ \Bigg( \frac{B_i}{B^2} \left(\sum_{m} \hat{y}_m + \sum_{m} \frac{b_m}{a_m}\right) +\\ - c_{i0}\frac{B_i}{B} - \frac{p^0}{N},  \frac{V_i B_i}{B^2} \Bigg)^{N-1}_{i=1} \Bigg\},
    \end{multline}
    and for node $0$ we can write it as follows:
    \begin{multline}\label{eq: pseudo-differential heterogeneous zero}
        \bm{F}(\bm{\hat{y}}, \bm{V}) = \Bigg( \frac{B_0}{B^2} \left(\sum_{m} \hat{y}_m + \sum_{m} \frac{b_m}{a_m}\right) +\\ + \sum_{m\not=0}c_{0m} \frac{B_m}{B} - \frac{p^0}{N},  \frac{V_0 B_0}{B^2} \Bigg).
    \end{multline}
    We want to prove that the operator $\bm{F}(\bm{\hat{y}}, \bm{V})$ defined in \eqref{eq: pseudo-differential}  or in \eqref{eq: pseudo-differential heterogeneous} and \eqref{eq: pseudo-differential heterogeneous zero}
    is strongly monotone.
    
    We denote vector $z$ to be $z:= (\hat{y}_0, V_0, \dots, \hat{y}_{N-1}, V_{N-1})$. We need to investigate whether $\nabla_{z} \bm{F}(\bm{\hat{y}}, \bm{V})$ is positive-definite. Denote $\bm{F}(\bm{\hat{y}}, \bm{V})_{i1} := \frac{B_i}{B^2} \left(\sum_{n} \hat{y}_m + \sum_{n} \frac{b_n}{a_n}\right) -c\frac{B_i}{B} - \frac{p^0}{N}$ for the homogeneous $c_{nm}$ case, and $\bm{F}(\bm{\hat{y}}, \bm{V})_{i2} = \frac{V_i B_i}{B^2} $. Similarly, for the heterogeneous $c_{nm}$ case, denote $\bm{F}(\bm{\hat{y}}, \bm{V})_{i1} := \frac{B_i}{B^2} \left(\sum_{n} \hat{y}_m + \sum_{n} \frac{b_n}{a_n}\right) -c_{i0}\frac{B_i}{B} - \frac{p^0}{N}$ for $i\not=0$ and $\bm{F}(\bm{\hat{y}}, \bm{V})_{01} := \frac{B_i}{B^2} \left(\sum_{n} \hat{y}_m + \sum_{n} \frac{b_n}{a_n}\right) + \sum_{m\not=0}c_{0m} \frac{B_m}{B} - \frac{p^0}{N}$. Thus we have that $\frac{\partial \bm{F}(\bm{\hat{y}}, \bm{V})_{i1}}{\partial \hat{y}_j} = \frac{B_i}{B^2}, \, \forall j \in \mathcal{N}$ and $\frac{\partial \bm{F}(\bm{\hat{y}}, \bm{V})_{i2}}{\partial V_i} = \frac{B_i}{B^2}$. All other partial derivatives are 0. Thus $\nabla_{z} \bm{F}(\bm{\hat{y}}, \bm{V})$ is a matrix defined with its entries to be 
    \begin{equation*}
        \nabla_{z} \bm{F}(\bm{\hat{y}}, \bm{V})_{ij} = \left\{
        \begin{aligned}
            &\frac{B_{\frac{i+2}{2}}}{B^2}  \text{ if $i,j$ are even } \\
            & \frac{B_{\frac{i+1}{2}}}{B^2} \text{ if $i,j$ are odd  and $i=j$} \\
            & 0 \text{ otherwise }
        \end{aligned}
        \right.
    \end{equation*}
    Symmetric matrix $A$ is positive definite on compact if its quadratic form is positive: $x^{\top}Ax > 0,\,\, \forall x \in \mathbb{R}^n \setminus 0$. Note, that non-symmetric matrix $A$ is positive definite iff symmetric matrix $\frac{1}{2}\left( A + A^{\top} \right)$ is. In our case the quadratic form is given by the following expression:
    \begin{multline*}
        \frac{1}{2} \bm{z}^{\top} \left( \nabla_{z} F(\bm{\hat{y}}, \bm{V}) + \nabla_{z} F(\bm{\hat{y}}, \bm{V})^{\top} \right)\bm{z} =\\
        = \frac {1}{B^2}\sum_{i =1 }^N B_i \hat{y}^2_i + \frac{1}{2B^2}\sum_{i=1}^N\sum_{j \not= i} (B_i + B_j) \hat{y}_i \hat{y}_j + \frac{1}{B^2}\sum_{i = 1}^N B_iV_i\\
        \geq \frac{1}{B^2}\sum_{i =1 }^N B_i \hat{y}^2_i + \frac{1}{B^2}\sum_{i=1}^N\sum_{j \not= i} \sqrt{B_i B_j} \hat{y}_i \hat{y}_j + \frac{1}{B^2}\sum_{i = 1}^N B_iV_i\\
        = \frac{1}{B^2}(\sum_{i=1}^N \sqrt{B_i}\hat{y}_i)^2 + \frac{1}{B^2}\sum_{i = 1}^N B_iV_i,
    \end{multline*}
    which is positive for all $(\hat{y}_i, V_i)_i$ in the feasible region. 
\end{proof}

\subsection{Proof of Proposition \ref{proposition: GPG}}
\begin{proof}
    The notion of Generalized Potential Game is widely used in the literature, see e.g. \cite{facchinei}. To verify that our game is indeed the potential game, consider function $\mathcal{P}(\bm{\hat{y}})$, defined as follows:
    \begin{equation*}
        \begin{aligned}
            \mathcal{P}(\bm{\hat{y}}, \bm{V}) &= \sum_{i =1}^N \Big[  \frac{H\hat{y}_i}{B}(\frac{1}{2}\sum_n \hat{y}_n + \sum_n \frac{b_n}{a_n}) -\frac{p^0 \hat{y}_i}{N}-\frac{cH\hat{y}_i}{B} \\
            &+ \frac{H}{B} V_i \Big]
        \end{aligned}
    \end{equation*}
    We can check that it is a potential function. Indeed, for all $\hat{y}_{-n}$, and for all admissible $x_n, z_n, x'_n, z'_n$:
    \begin{multline}
            \Pi_n(x_n,\hat{\bm{y}}_{-n}, x'_n, \bm{V}_{-n}) - \Pi_n(z_n,\hat{\bm{y}}_{-n},z'_n, \bm{V}_{-n}) = \\ 
            = \frac{H}{2B} \Big[ (x_n - z_n)(x_n+z_n+2\sum_{m \not= n} \hat{y}_m + 2\sum_{k \in \mathcal{N}}\frac{b_k}{a_k}) \Big] \\
            - \frac{(x_n - z_n)cH}{B}  - \frac{p^0}{N}(x_n - z_n)+ \\
            +\frac{H}{B} (x'_n - z'_n)=\\
            = P(x_n,\hat{\bm{y}}_{-n}, x'_n, \bm{V}_{-n}) - P(z_n,\hat{\bm{y}}_{-n},z'_n, \bm{V}_{-n})
    \end{multline}
\end{proof}
\end{document}